\documentclass[11pt,doublespace]{article}
\usepackage{amsfonts}
\usepackage{amsmath}
\usepackage{graphicx}
\usepackage{algorithmic}
\usepackage{algorithm}
\usepackage{MnSymbol}
\title{Weighted Elastic Net Penalized Mean-Variance Portfolio Design and Computation}

\author{Michael Ho\footnotemark[2] \and Zheng Sun\footnotemark[3] \and Jack Xin\footnotemark[4]}

\begin{document}
\maketitle
\newcommand{\slugmaster}{%
\slugger{sifin}{xxxx}{xx}{x}{x--x}}
\newenvironment{Prog}{\refstepcounter{Prog}\align}{\tag{P\theProg}\endalign}
\newenvironment{Prog*}{\align}{\endalign}
\renewcommand{\thefootnote}{\fnsymbol{footnote}}
\footnotetext[2]{Department of Mathematics, UC Irvine, Irvine, CA 92697, USA. Email: mtho1@uci.edu.}
\footnotetext[3]{The Paul Merage School of Business, UC Irvine, Irvine, CA 92697. Email: zsun@merage.uci.edu.}
\footnotetext[4]{Department of Mathematics, UC Irvine, Irvine, CA 92697, USA. Email: jxin@math.uci.edu.}
\footnotetext[5]{The work was partially supported by NSF grant DMS-1211179.}

\renewcommand{\thefootnote}{\arabic{footnote}}
\newcommand{\overbar}[1]{\mkern 1.5mu\overline{\mkern-1.5mu#1\mkern-1.5mu}\mkern 1.5mu}
\newcommand{\Ell}[1]
{  \ell_{#1}  }
\newcommand{\supp}[1]
{  \textnormal{supp}({#1})  }
\newcommand{\suppEpp}[2]
{  \textnormal{supp}_{#2}({#1})  }
\newcommand{\sign}[1]
{  \textnormal{sgn}({#1})  }
\newcommand{\Elln}[2]
{  ||#1||_{\ell_{#2}}  }
\newcommand{\EllnW}[3]
{  ||#1||_{#2,\ell_{#3}}  }

\newcommand{\EllnBig}[2]
{  \Big|\Big|#1 \Big|\Big|_{\ell_{#2}}  }
\newcommand{\E}
{ \mathbb{E} }
\newcommand{\seq}[2]
{\{#1_{n}\}_{n=1}^{#2} }
\newcommand{\Ft}
{\mathcal{F}_{t}}
\newcommand{\Yt}
{Y_{t}}
\newcommand{\hYt}
{\hat{Y}_{t}}
\newcommand{\SigEp}
{\Sigma^{\epsilon}}
\newcommand{\wTil}
{\tilde{w}}
\newcommand{\Rtil}
{\tilde{R}}
\newcommand{\vect}[1]
{\vec{#1}}
\newcommand{\partialDer}[1]
{ \frac{\partial}{\partial #1} }
\newtheorem{theo}{Theorem}
\newtheorem{theorem}[theo]{Theorem}
\newtheorem{lemma}[theo]{Lemma}
\newtheorem{prop}[theo]{Proposition}
\newtheorem{corr}[theo]{Corollary}
\newtheorem{proof}{Proof}
\newtheorem{Def}{Definition}
\newtheorem{Remark}{Remark}
\newtheorem{definition}{Definition}
\newtheorem{Conj}{Conjecture}
\newtheorem{example}{Example}

\begin{abstract}
It is well known that the out-of-sample performance of Markowitz's mean-variance portfolio criterion can be negatively affected by estimation errors in the mean and covariance.  In this paper we address the problem by regularizing the mean-variance objective function with a weighted elastic net penalty.   We show that the use of this penalty can be motivated by a robust reformulation of the mean-variance criterion that directly accounts for parameter uncertainty.   With this interpretation of the weighted elastic net penalty we derive data driven techniques for calibrating the weighting parameters based on the level of uncertainty in the parameter estimates.  We test our proposed technique on US stock return data and our results show that the calibrated weighted elastic net penalized portfolio outperforms both the unpenalized portfolio and uniformly weighted elastic net penalized portfolio.

This paper also introduces a novel Adaptive Support Split-Bregman approach which leverages the sparse nature of $\ell_{1}$ penalized portfolios to efficiently compute a solution of our proposed portfolio criterion.  Numerical results show that this modification to the Split-Bregman algorithm results in significant improvements in computational speed compared with other techniques.\\
{\tiny This article is to appear in SIAM J. Financial Mathematics.}
\end{abstract}

\pagestyle{myheadings}
\thispagestyle{plain}
\markboth{M. HO, Z. SUN, AND J. XIN}{WEIGHTED ELASTIC NET PENALIZED PORTFOLIOS}

\section{Introduction}
  The birth of modern portfolio theory occurred in 1952 with the seminal publication of Harry Markowitz's criterion \cite{Markowitz} for optimal single period portfolio construction that balances a portfolio's risk with return potential.  A key assumption in modern portfolio theory is that given two portfolios with the same expected return an investor will always choose the portfolio with minimal risk.   Markowitz proposed using the variance of portfolio's return as the measure of the portfolio risk.   Thus Markowitz formulated the portfolio selection problem as minimizing portfolio return variance subject to a minimum expected value of return.   Mathematically the Markowitz formulation can be written as a quadratic programming problem and the optimal portfolio can be computed using a variety of quadratic programming methods \cite{BoydConvex,NocedalBook}.

One shortcoming of the Markowitz criterion for portfolio optimization is that it requires the practitioner to specify the expected return of each asset and the covariance of the returns of different assets.  This presents a problem to an investor because the future mean and covariance matrix are not known.  If incorrect parameter values are used then the portfolio performance will be sub-optimal \cite{MichaudEnigma,DeMiguelOneOverN}.  This additional risk due to parameter uncertainty is commonly referred to as estimation risk.

An intuitive technique that can be utilized when the mean and covariance are unknown is to estimate the mean and covariance matrix from historical return data \cite{LittermanCovar} and to plug-in the estimated parameters in place of the truth.   One approach to estimating the unknown parameters is to use sample averaging which is maximum likelihood (ML) optimal when the returns are i.i.d and normally distributed.   This approach can be very accurate when the data is normally distributed and sufficient training data is available.   For data that is not normally distributed robust estimation techniques for the covariance matrix can be considered \cite{PalomarTyler,HeroRobust}.

Although the sample average and plug-in approach is intuitive, there are difficulties in effectively implementing it.  The primary difficulty is that there is often a limited amount of relevant historical financial return data available to estimate the mean and covariance.   One reason for the lack of relevant data is that the investments' return statistics can be time-varying.   Thus only a limited amount of past data is relevant in estimating the current mean and covariance.   Since the volatility of assets returns can be large, the parameter estimates obtained from averaging only a small number of samples can be large.   Further complicating the problem is that the covariance matrix can be ill-conditioned.   This makes the portfolio weights extremely sensitive to small parameter errors.    The effect of these estimation errors is risk return performance that departs significantly from the optimal performance under known statistics \cite{DeMiguelOneOverN,Barry1974,Jobson1980}.

As an alternative to sample average estimates, Bayesian estimators for both mean and covariance have been proposed \cite{Frost1986,Jorion1986,LedoitWolf2004}.   These estimators effectively ``shrink'' the sample average estimates towards a more structured estimate (via a convex combination) which takes into account prior knowledge.   Prior knowledge can take the form of structured data models such as a single factor model \cite{Sharpe1963} or the Fama- French three-factor model \cite{FamaFrench1993}.   Shrinking the sample average estimates towards the more structured model reduces the variability in the parameter estimates and can improve out-of-sample portfolio performance.


Another approach that has been shown to improve out-of-sample performance involves regularizing the portfolio selection criterion by adding penalties to the objective function \cite{DeMiguelPORTNORM,SparseStable,OptimalSparse,YenYen} such as portfolio norm penalties.   In \cite{DeMiguelPORTNORM} $\ell_{1}$ and squared $\ell_{2}$ norm constraints are proposed for the minimum variance criterion and a cross-validation procedure is suggested to calibrate the constraints.   In \cite{YenYen,L1L2_london} an elastic net penalty is proposed in the context of constrained minimum variance portfolio optimization.   The authors also derive a method to calibrate the elastic net penalty which is designed to ensure that the variance of the resulting portfolio will not exceed the unpenalized portfolio variance (asymptotically).   In \cite{OptimalSparse} the authors study convex penalties such as a weighted LASSO approach \cite{AdaptiveLasso} and a non-convex SCAD penalty \cite{SCAD} with application to minimum-variance portfolios.   For the weighted LASSO approach the authors propose a calibration scheme where the weights are selected according to the variability in the volatility of each asset.

The above norm constrained and penalty approaches for portfolio optimization primarily focused on the minimum variance approach.   Consequently the calibration and justification of the norm penalties above were derived by considering only the portfolio variance.   The mean return is ignored in the calibration of the penalty.   In this paper we propose a method which can be applied to the mean-variance criterion where both portfolio mean and variance are considered.   In this setting we propose regularizing the objective function with a weighted elastic net penalty.   A weighted elastic net penalty is a linear combination of a portfolio's weighted $\ell_{1}$ norm and the square of a portfolio's weighted $\ell_{2}$ norm.   We show that the use of the weighted elastic net penalty can be justified by reformulating the mean-variance criterion as a robust optimization problem \cite{GoldFarbIyengar,TutKoenig} where the mean and volatilities of the asset returns belong to a known uncertainty set.  With this robust optimization interpretation, data driven techniques for calibrating the weight parameters in the weighted elastic net penalty are derived.

Our proposed penalized criterion which is equivalent to a special case of a robust optimization problem has two advantages over the general robust portfolio optimization problem.   First our method can be solved using fast and well established algorithms for $\ell_{1}$ penalized optimization problems such as the Split-Bregman algorithm \cite{GoldStein1} and the FISTA algorithm \cite{FISTA}.   In the more general case, solving the robust portfolio optimization problem requires using semi-definite programming techniques \cite{InteriorSPD} which are intractable for large portfolios.   Finally, our formulation of the problem results in sparse portfolios which can contribute to reducing portfolio turnover and transaction costs.   The general robust optimization problem does not necessarily result in a sparse portfolio.

This paper also addresses computational aspects of computing weighted elastic net penalized portfolios.   In particular, we propose a novel Adaptive Support Split-Bregman approach to computing weighted elastic-net penalized portfolios.  This new algorithm exploits the sparse nature of elastic net penalized solutions to minimize computational requirements.   We show that this results in significant improvements in convergence speed versus other solvers.

The remainder of this paper's body is organized as follows:  Section 2 introduces the weighted elastic net penalty and provides a justification for its use.   In Section 3 we discuss the Adaptive Support Split-Bregman approach for computing the optimal portfolio.   Section 4 presents experimental results using US equity data that demonstrate the benefit of our proposed approach.   Finally in Section 5 we state our conclusions and a path forward for future work.   The appendix contains proofs of some technical results presented in Section 3.

\section{Portfolio Selection Criteria}\label{Sec:Crit}
In this section we first review the mean-variance portfolio selection criterion.   We then present the weighted elastic-net penalized portfolio selection criterion and motivate its use.
\subsection{Mean-Variance Portfolio selection criterion}   Suppose that there exists a set of $N$ risky assets and let $\{s_{n}(k)\}_{n=1}^{N}$ be the prices of each asset at time $k$.   Then the excess return of the $n^{th}$ asset for time period $k$ is defined as
\begin{equation}
  r_{n}(k)=\frac{s_{n}(k+1) - s_{n}(k)}{s_{n}(k)} -r^{(F)}(k)
\end{equation}
where $r^{(F)}(k)$ is the return of a risk-free asset at time $k$.
We model $\{r_{n}\}_{n=1}^{N}$ as random variables with finite mean and covariance.
A portfolio is defined to be a set of weights $\{w_{n}\}_{n=1}^{N} \subset \mathbb{R}$.  If $w_{i}>0$ a long position has been taken in the $i^{th}$ asset whereas $w_{i}<0$ indicates a short position.

The mean-variance criterion proposed by Markowitz \cite{Markowitz} addresses single period portfolio selection.  A portfolio of risky assets, $w$ is mean-variance optimal if it is a solution to the following optimization problem
\begin{align}\label{eq:MeanVar}
   \min_{w}& \;\; \varphi w^{T}\Gamma w - \mu^{T}w
\end{align}
where $\Gamma$ and $\mu$ are the covariance and mean of $r$ for the time period of interest and where $\varphi>0$ is a risk aversion coefficient (since $\varphi$ will only affect the portfolio weights up to a positive scalar multiple we shall set $\varphi=1$) .      Assuming that $\Gamma$ is symmetric and positive definite we have that \eqref{eq:MeanVar} is a convex quadratic program whose solution, $w^{*}$, satisfies
\begin{equation}
  \Gamma w^{*}=\mu
\end{equation}

Estimation of parameters is necessary to implement the mean-variance criteria.    It has been recognized that estimation of mean return is more difficult than covariance \cite{MertonMean} and thus a minimum variance criterion is often advocated for in recent literature \cite{Ma,DeMiguelPORTNORM,OptimalSparse}.   In the minimum variance criterion the mean of asset returns are ignored and the following criterion is used for portfolio selection
\begin{align}
   \min_{w}& \;\;  w^{T}\Gamma w \nonumber \\
   & \textnormal{s.t. }  \sum_{i=1}^{N}w_{i} = 1.
\end{align}
Despite ignoring all information on the mean return, the minimum variance criterion often outperforms the mean-variance criterion when judged by out-of-sample Sharpe ratio \cite{DeMiguelPORTNORM,Ma}.

\subsection{Norm Penalized Portfolio optimization}
As was stated in the introduction mean-variance portfolio optimization is sensitive to parameter estimation error.   To address these concerns a number of norm penalized criterions have been proposed, primarily in the context of minimum variance optimization.   Commonly used convex norm penalties include the $\ell_{1}$ norm , squared $\ell_{2}$ norm and elastic net penalties \cite{YenYen}.   The $\ell_{1}$ and squared $\ell_{2}$ norm penalties are given as
\begin{equation}
\lambda \sum_{i=1}^{N}|w_{i}|
\end{equation}
and
\begin{equation}
\lambda \sum_{i=1}^{N}w_{i}^{2}
\end{equation}
respectively where $\lambda>0$ is a weighting factor.   The elastic net penalty is a weighted sum of the $\ell_{1}$ and squared $\ell_{2}$ norm penalties
\begin{equation}\label{eq:ElasticEq}
\lambda_{1} \sum_{i=1}^{N}|w_{i}| +\lambda_{2} \sum_{i=1}^{N}w_{i}^{2}
\end{equation}
where $\lambda_{1},\lambda_{2} >0$.
Another convex penalty is the adaptive LASSO penalty \cite{AdaptiveLasso} which was applied to minimum variance optimization in \cite{OptimalSparse}.   The adaptive LASSO penalty is a weighted $\ell_{1}$ norm given by
 \begin{equation}
   \EllnW{w}{\vect{\beta}}{1} = \sum_{k=1}^{N} \beta_{k}|w_{k}|
 \end{equation}
where $\beta_{k}\ge 0$.
Calibration of the weighting parameters for the above penalties has primarily been studied with the goal of improving the portfolio return variance \cite{YenYen,OptimalSparse}.

Several justifications for using $\ell_{1}$ and squared $\ell_{2}$ norms as penalties and constraints have been given in the literature.   For example in \cite{SparseStable} it is stated that the use of an uniformly weighted $\ell_{1}$ penalty can be motivated by the desire to obtain sparse portfolios and to regularize the mean-variance problem when the covariance is ill-conditioned.    In \cite{VastPortfolio} the authors show that estimation risk in the mean-variance setting due to errors in the mean return estimation is bounded above by
\begin{eqnarray}
  ||\mu-\hat{\mu}||_{\infty}\Elln{w}{1}
\end{eqnarray}
and use that upper bound as a rationale for promoting small $\Elln{w}{1}$.   In \cite{SparseStableParameter} it is mentioned that a benefit of using a uniformly weighted $\ell_{2}$ norm penalty is to stabilize the inverse covariance matrix which is often ill-conditioned in financial applications.

Non-convex penalized minimum-variance portfolio criterions were studied in \cite{OptimalSparse}.   One such penalty examined in \cite{OptimalSparse} is the Softly Clipped Absolute Deviation (SCAD) penalty \cite{SCAD}.  The SCAD penalty is defined as follows
\begin{equation}
\sum_{i=1}^{N} p_{\lambda}(w_{i})
\end{equation}
where
\begin{equation}\label{eqSCAD}
p_{\lambda}(x) = \begin{cases} \lambda|x| &\mbox{if } |x| \le \lambda \\
-\frac{x^{2}-2a_{SCAD}\lambda|x|+\lambda^{2}}{2(a_{SCAD}-1)} &\mbox{if } \lambda <|x| \le a_{SCAD}\lambda \\
\frac{(a_{SCAD}+1)\lambda^{2}}{2} &\mbox{if }  |x|> a_{SCAD}\lambda
\end{cases}
\end{equation}
and where $a_{SCAD} >2$.   This penalty is similar to the $\ell_{1}$ penalty and was initially proposed in context of variable selection.   Calibration of the parameters $a_{SCAD}$ and $\lambda$ in \eqref{eqSCAD} for portfolio optimization has not been fully addressed in the literature.

\subsection{Weighted Elastic Net Penalized Portfolio}\label{sec:WtElastic}
The preceding norm penalties are derived and calibrated primarily from a minimum variance perspective.  In this section we extend the above methods for minimum variance portfolio design to mean-variance portfolios.   Here we propose augmenting the mean-variance criterion with the sum of a weighted $\ell_{1}$ and the square of a weighted $\ell_{2}$ penalty.   The penalty terms in the new portfolio selection criterion will be referred to as a weighted elastic net which was studied in the context of variable selection in \cite{AdaptiveElasticNet}.

Let $\left\{\alpha_{i}\right\}_{i=1}^{N}$ and $\left\{\beta_{i}\right\}_{i=1}^{N}$ be positive real numbers.   Then the weighted elastic net penalty for a portfolio $w$ is
\begin{equation}
\EllnW{w}{\vect{\beta}}{1}+ \EllnW{w}{\vect{\alpha}}{2}^{2}
\end{equation}
 where
 \begin{equation}
   \EllnW{w}{\vect{\beta}}{1} = \sum_{k=1}^{N} \beta_{k}|w_{k}|
 \end{equation}
 and
 \begin{equation}
   \EllnW{w}{\vect{\alpha}}{2}^{2} = \sum_{k=1}^{N} \alpha_{k}|w_{k}|^{2}.
 \end{equation}
 Thus the weighted elastic net penalized mean-variance criterion may be written as
\begin{eqnarray}\label{Prog:WtElasticNet}
  \min_{w}& \;\; w^{T}\hat{\Gamma} w -w^{T}\hat{\mu}+\EllnW{w}{\vect{\beta}}{1}+ \EllnW{w}{\vect{\alpha}}{2}^{2}
\end{eqnarray}
where $\hat{\Gamma}$ and $\hat{\mu}$ are estimates of $\Gamma$ and $\mu$ respectively.
\subsection{Motivation}

 A rationale for augmenting the mean-variance criterion with a weighted elastic net penalty can be obtained by reformulating the mean-variance criterion as a robust optimization problem.   As was stated in the introduction it is well-known that the out-of-sample performance of mean-variance portfolio can degrade significantly when there are errors in the estimate of mean and covariance.   The risk due to estimation errors can be reduced by accounting for them in the optimization criterion.

One way to model the parameter estimation risk is to assume the true covariance and mean belong to the following uncertainty sets
\begin{eqnarray*}
  A = \left\{R : R_{i,j}=\hat{\Gamma}_{i,j}+e_{i,j}
   ;  |e_{i,j}| \le \Delta_{i,j}; R \succeq 0 \right\}
\end{eqnarray*}
\begin{equation*}
  B =\left\{v : v_{i}= \hat{\mu}_{i}+c_{i}; |c_{i}| \le \beta_{i} \right\}
\end{equation*}
where the matrix $\Delta$ is symmetric and diagonally dominant with $\Delta_{i,j} \ge 0$ for all $i,j$.   This condition on $\Delta$ ensures that a matrix, $R$, of the form
\begin{eqnarray*}
  R_{i,j}&=& \begin{cases} \hat{\Gamma}_{i,i} + \Delta_{i,i} \mbox{ if } i = j \\
  \hat{\Gamma}_{i,j} \pm  \Delta_{i,j} \mbox{ if } i \ne j
  \end{cases} \\
  R_{i,j}&=&R_{j,i}
\end{eqnarray*}
is positive semi-definite (i.e. $R \in A$).

Since the mean and covariance are unknown, a conservative approach to selecting a portfolio is to optimize the worse case performance over the uncertainty sets.   This can be written as the following robust optimization problem \cite{GoldFarbIyengar}
\begin{eqnarray} \label{Prog:Robust}
  \min_{w} \max_{R \in A, v \in B} w^{T}Rw -v^{T}w.
\end{eqnarray}
Note that for a fixed $R$ and $v$ this problem is convex in $w$.  Since the pointwise maximum of a family of convex functions remains convex we have that
\begin{eqnarray}
\max_{R \in A, v \in B} w^{T}Rw -v^{T}w
\end{eqnarray}
is convex in $w$.  Performing the inner maximization with respect to $\mu$ reduces the problem to
\begin{eqnarray} \label{Prog:Robust2_p1}
  \min_{w} \max_{R \in A} w^{T}Rw + \sum_{i=1}^{N} \left( -\hat{\mu_{i}}+ \beta_{i}\textnormal{sgn}({w_{i}}) \right)w_{i}
\end{eqnarray}
where
\begin{eqnarray*}
  \textnormal{sgn}(w_{i}) =
  \begin{cases} \frac{w_{i}}{|w_{i}|} & \mbox{ if } w_{i} \ne 0 \\
  0 &\mbox {else}.
  \end{cases}
\end{eqnarray*}
This can be re-written as
\begin{eqnarray*}
  \min_{w} \max_{R \in A} \textnormal{  tr}(Rww^{T}) -w^{T}\hat{\mu}  + \EllnW{w}{\vect{\beta}}{1},
\end{eqnarray*}
and the inner maximization with respect to $R$ can be solved in closed form.   Performing this final maximization gives us the following convex optimization problem
\begin{eqnarray}\label{Prog:Pairwise}
  \min_{w} w^{T}\hat{\Gamma}w  -w^{T}\hat{\mu} +|w|^{T}\Delta|w| + \EllnW{w}{\vect{\beta}}{1}
\end{eqnarray}
where the $N \times 1$ vector $|w|$ is defined as
\begin{eqnarray}
  |w|_{i}=|w_{i}|.
\end{eqnarray}
Thus we see that problem \eqref{Prog:Robust} is equivalent to augmenting the mean-variance criterion with a weighted pairwise elastic net penalty \cite{PairwiseElastic}.

When $\Delta$ equals the diagonal matrix $D_{\alpha}$ where
\begin{eqnarray}
D_{\alpha} = \left( \begin{matrix}
               \alpha_{1} & 0 & \dots & 0 \\
               0 & \ddots & \ddots & \vdots \\
               \vdots & \ddots  & \ddots & 0 \\
               0 & \dots & 0 &\alpha_{N}  \\
             \end{matrix} \right)
\end{eqnarray}
the criterion simplifies to the weighted elastic net penalized problem defined in problem \eqref{Prog:WtElasticNet}
\begin{eqnarray}
  \min_{w} w^{T}\hat{\Gamma}w -w^{T}\hat{\mu}  + \EllnW{w}{\vect{\beta}}{1} + \EllnW{w}{\vect{\alpha}}{2}^{2}
\end{eqnarray}
where $\alpha_{i}=\Delta_{i,i}$.  This observation is summarized in the following theorem:
\begin{theorem}\label{Thm:Equiv}
The weighted elastic net penalized problem in \eqref{Prog:WtElasticNet} is equivalent to the robust optimization problem in \eqref{Prog:Robust}, when $\Delta=D_{\alpha}$.
\end{theorem}
\subsection{Data Driven Calibration of Weighting Parameters}\label{SEC:Cal}
We now address the problem of selecting the weighting parameters $\alpha$ and $\beta$.   Recall that Theorem \ref{Thm:Equiv} states that problems \eqref{Prog:WtElasticNet} and \eqref{Prog:Robust} are equivalent.   This implies that $\alpha$ and $\beta$ represent the level of uncertainty in the mean and variance of each asset.   Thus we propose setting $\alpha$ and $\beta$ to be commensurate with the amount of error in our parameter estimates.

Since the amount of error in the parameter estimates are unknown, we need to estimate them.  One approach to estimate the amount of error is the bootstrap method \cite{EfronBoot}.   Bootstrapping is a non-parametric approach that has been applied to portfolio optimization \cite{MichaudBook} and calibration of robust portfolio optimization problems \cite{TutKoenig}.   One advantage of bootstrapping is that it does not require specification of a distribution of the return data.

The first step of bootstrapping is to measure $T_{train}$ time samples of past training data to estimate $\mu_{i}$ and $\Gamma_{i,i}$, using estimators $f_{\mu_{i}}$ and $f_{\Gamma_{i,i}}$ respectively.   Common choices for $f_{\mu_{i}}$ and $f_{\Gamma_{i,i}}$ are sample averages or shrinkage estimators.   Once the parameter estimates are obtained, the training data is resampled with replacement and additional estimates of $\mu_{i}$ and $\Gamma_{i,i}$ are formed using the resampled data.  The resampling can be described by independent uniformly distributed integer valued random variables, $v_{k,m}$, taking values between 1 and $T_{train}$.   Here $k \in \{1, \dots, K\}$ and $m \in \{1, \dots, T_{train}\}$.   Under the condition that the estimators $f_{\mu_{i}}$ and $f_{\Gamma_{i,i}}$ are invariant to the ordering of the training data, the bootstrap estimates of the estimation errors may be defined as
\begin{equation*}
  \mu_{i,err}(k) = \left| f_{\mu_{i}}(r_{i}(v_{k,1}), \dots, r_{i}(v_{k,T_{train}}))  -\hat{\mu}_{i} \right|
\end{equation*}
and
\begin{equation*}
  \Gamma_{i,err}(k) = \left| f_{\Gamma_{i,i}}(r_{i}(v_{k,1}), \dots, r_{i}(v_{k,T_{train}})) -\hat{\Gamma}_{i,i}  \right|
\end{equation*}
respectively.    Here $r_{i}(t)$ is the return of the $i^{th}$ asset in the $t^{th}$ training sample.  The percentiles of the empirical distributions of $\{\Gamma_{i,err}(k)| k=1 \dots K \}$ and $\{\mu_{i,err}(k)| k=1 \dots K \}$ can then be referenced to derive $\alpha_{i}$ and $\beta_{i}$.  For example, suppose $0 \le p_{1},p_{2} \le 1$.   Then the values for $\alpha_{i}$ and $\beta_{i}$ can be defined as
\begin{equation}
  \alpha_{i}= \min \left\{x: \left| \left\{n: \Gamma_{i,err}(n) \le x \right\} \right|\le p_{1}K \right\}
\end{equation}
and
\begin{equation}
  \beta_{i}= \min \left\{x: \left| \left\{n: \mu_{i,err}(n) \le x \right\} \right|\le p_{2}K \right\}
\end{equation}
where $K$ is the number of bootstrap estimates.

An economic interpretation of the percentile parameters $p_{1}$ and $p_{2}$ is that of model estimation risk aversion factors.   Here $p_{1}$ represents the aversion to squared volatility estimation risk and $p_{2}$ is the aversion to mean estimation risk.   A percentile value of 0 corresponds to no aversion to estimation risk whereas a value of 1 corresponds to a high aversion to estimation risk.   Note that a higher aversion to estimation risk will increase the weights in the elastic net.

%
%
%

\section{Numerical methods}
In this section we review some numerical algorithms for determining solutions of \eqref{Prog:WtElasticNet}.  First we review an application of the Split-Bregman algorithm \cite{GoldStein1} for solving \eqref{Prog:WtElasticNet}.   Then we propose a novel Adaptive Support Split-Bregman approach which solves \eqref{Prog:WtElasticNet} faster than the Split-Bregman algorithm by exploiting the sparse nature of the portfolio weights.
\subsection{Optimality and Approximate Optimality Conditions}
In this section we derive approximate optimality conditions for \eqref{Prog:WtElasticNet}.   These conditions are then used to a design a numerical algorithm for determining the solution of \eqref{Prog:WtElasticNet}.

Let $\Psi(w)$ denote the objective function for the weighted elastic net portfolio problem in equation \eqref{Prog:WtElasticNet}
\begin{eqnarray}
\Psi(w)&=&w^{T}\hat{\Gamma} w -w^{T}\hat{\mu}+\EllnW{w}{\vect{\beta}}{1}+ \EllnW{w}{\vect{\alpha}}{2}^{2} \\
&=&w^{T}Rw -w^{T}\hat{\mu}+\EllnW{w}{\vect{\beta}}{1} \nonumber
\end{eqnarray}
where $R=\hat{\Gamma}+ D_{\alpha}$.   Since $\Psi$ is convex, $w^{*}$ minimizes $\Psi$ if and only if
\begin{align}
0 \in \partial\Psi(w^{*})
\end{align}
where $\partial\Psi(w)$ is the sub-gradient of $\Psi$ evaluated at $w$ \cite{BoydConvex}.   Note that since $R$ is positive definite, $\Psi$ is strictly convex and thus there is a unique solution to \eqref{Prog:WtElasticNet}.

In most cases we are only interested in portfolios that are approximately optimal.  Thus we can relax our optimality conditions to derive a stopping criterion for any iterative solver of \eqref{Prog:WtElasticNet}.  Before introducing our relaxed conditions we define the support of a portfolio $w$ as
\begin{equation*}
  \supp{w} = \left\{i : |w_{i}| >  0 \right\}
\end{equation*}
and define the smallest variance uncertainty as
\begin{equation}
\alpha_{o}=\min\{\alpha_{i}:  0 \le i \le N\}.
\end{equation}
With the above definitions we have the following theorem which establishes an approximate optimality condition.
\begin{theorem}\label{THM:ConvCriteria1}
Let $w^{*}$ be the solution of \eqref{Prog:WtElasticNet}.   Suppose that $\tilde{w}$ satisfies
\begin{equation}\label{EQ:gradCondition}
    \sum_{i \in \supp{\tilde{w}}}\left(\frac{\partial}{\partial w_{i}} \big( w^{T}Rw  -w^{T}\hat{\mu} +\EllnW{w}{\vect{\beta}}{1} \big)\Big\mid_{w=\tilde{w}}\right)^{2} \le 2\epsilon \alpha_{o}
\end{equation}
and
\begin{equation}\label{EQ:ABSineq}
   -\beta_{i} \le \frac{\partial}{\partial w_{i}} \big( w^{T}Rw  -w^{T}\hat{\mu} \big)\Big\mid_{w=\tilde{w}} \le \beta_{i}
\end{equation}
for all $i \nin \supp{\tilde{w}}$.   Then
\begin{equation}\label{EQ:PhiIneq}
   \Psi(\tilde{w}) \le \Psi(w^{*}) +\epsilon
\end{equation}
\end{theorem}
\begin{proof}
See Appendix.
\end{proof}

In a numerical algorithm it may happen that none of the portfolio weights are exactly 0, although they may be extremely close to zero.  Thus the above theorem may not be very practical for use as a stopping criterion.  For this reason let us separate the small portfolio weights from the larger portfolio weights.  To do this we define
\begin{equation*}
  \suppEpp{w}{\epsilon} = \left\{i \in \supp{w} : |w_{i}| < \epsilon \right\}.
\end{equation*}
With this definition we have the following corollary which suggests a more practical stopping rule than Theorem \ref{THM:ConvCriteria1}.
\begin{theorem}\label{THM:ConvCriteria2}
Let $M \ge 2\Elln{R}{2}$ and let $\epsilon >0$ be given.  Choose $\eta < \frac{\epsilon\wedge \sqrt{\epsilon \alpha_{o}}}{\sqrt{N}M}$.   Let $w^{*}$ be the solution the of \eqref{Prog:WtElasticNet}.   Suppose that $\tilde{w}$ satisfies
\begin{equation}\label{EQ:gradConditionPractA}
    \sum_{i \in \; \supp{\tilde{w}}\setminus \suppEpp{\tilde{w}}{\eta}}\left(\frac{\partial}{\partial w_{i}} \big( w^{T}Rw  -w^{T}\hat{\mu} +\EllnW{w}{\vect{\beta}}{1} \big)\Big\mid_{w=\tilde{w}}\right)^{2} \le 2\epsilon \alpha_{o}
\end{equation}
and
\begin{equation}\label{EQ:gradConditionPractB}
   -\beta_{i} +\epsilon  \le \frac{\partial}{\partial w_{i}} \big( w^{T}Rw  -w^{T}\hat{\mu} \big)\Big\mid_{w=\tilde{w}} \le \beta_{i} -\epsilon
\end{equation}
for $i \in \suppEpp{\tilde{w}}{\eta} \cup \overbar{\supp{\tilde{w}}}$.  Then
\begin{equation}\label{EQ:PhiIneqPract}
   \Psi(\zeta) \le \Psi(w^{*}) +\frac{(\sqrt{2}+1)^{2}}{2}\epsilon
\end{equation}
where
 \begin{equation*}
   \zeta_{i}=\begin{cases}
   0 \mbox{ if }  i \in \suppEpp{\tilde{w}}{\eta} \\
   \tilde{w}_{i} \mbox { else .} \\
   \end{cases}
 \end{equation*}
\end{theorem}
\begin{proof}
See Appendix.
\end{proof}

\subsection{Split-Bregman Algorithm}
The weighted elastic net problem can be reformulated as a quadratic program and solved using general purpose solvers.  However the reformulation involves adding an additional $N$ primal variables as well as $2N$ dual variables.   Thus this approach may not be applicable to large scale problems.

An algorithm better suited to handle problems like \eqref{Prog:WtElasticNet} is the Split-Bregman algorithm.
The Split-Bregman algorithm was introduced in \cite{GoldStein1} for problems involving $\ell_{1}$ regularization such as \eqref{Prog:WtElasticNet}.
When using the Split-Bregman method to solve \eqref{Prog:WtElasticNet} we solve an equivalent problem
\begin{align}\label{SBreform1}
  \min_{w,d}& \;\; w^{T}R w -w^{T}\hat{\mu}+||d||_{\Ell{1}} \nonumber \\
  \textnormal{s.t.     }& d=\psi(w)
\end{align}
where $R=\rho \hat{\Gamma}+ D_{\alpha}$ and where $\psi(w)=(\beta_{1}w_{1}, \dots, \beta_{N}w_{N} )$.
The Split-Bregman algorithm applied to \eqref{SBreform1} is
\begin{algorithm}[H]
\begin{algorithmic}
\caption{Split Bregman Algorithm for solving \eqref{SBreform1}}
\label{ALGSplitGeneralized}
\STATE
\STATE  \textbf{Initialize: }$k=1$, $b^{k}=0,w^{k}=0, d^{k}=0 $
\WHILE{$||w^{k}-w^{k-1}||_{\Ell{2}}>tol$}
\STATE $w^{k+1}=\arg \min_{w} w^{T}Rw -w^{T}\hat{\mu} + \frac{\lambda}{2}||d^{k}- \psi(w) - b^{k}||_{\Ell{2}}^{2}$
\STATE $d^{k+1}=\arg \min_{d}  \frac{\lambda}{2}||d- \psi(w^{k+1}) - b^{k}||_{\Ell{2}}^{2}+||d||_{\Ell{1}}$
\STATE $b^{k+1}_{i}=b^{k}_{i} + \beta_{i} w^{k+1}_{i} - d^{k+1}_{i}$
\STATE $k=k+1$
\ENDWHILE
\end{algorithmic}
\end{algorithm}

Both inner optimization problems in Algorithm \ref{ALGSplitGeneralized} have closed form solutions.   The first problem is an unconstrained strictly convex quadratic program and the second problem can be solved using the shrinkage operator
\begin{equation*}
d^{k+1}_{j}=shrink(\beta_{j} w_{j}^{k+1} +b_{j}^{k},\frac{1}{\lambda})
\end{equation*}
where
\begin{equation*}
shrink(x,\gamma)=\frac{x}{|x|} \cdot \max (|x|-\gamma,0).
\end{equation*}

The stopping criterion in Algorithm \ref{ALGSplitGeneralized} does not ensure that the objective value is within a desired tolerance.    A modification to the algorithm can be made to ensure that this occurs.   One such modification uses Theorem \ref{THM:ConvCriteria2} to derive a stopping criterion.
\begin{algorithm}[H]
\begin{algorithmic}
\caption{Modified Split Bregman Algorithm for solving \eqref{SBreform1}}
\label{ALG:ModifiedSplit}
\STATE
\STATE  \textbf{Initialize: }$k=0$, $b^{k},w^{k}, d^{k}=|w^{k}|, tol>0 $
\WHILE{$w^{k}$ does not satisfy conditions of Theorem \ref{THM:ConvCriteria2} for $\epsilon =\frac{2}{(\sqrt{2}+1)^{2}}tol$ and $\tilde{w}=w^{k}$}
\STATE $w^{k+1}=\arg \min_{w} w^{T}Rw -w^{T}\hat{\mu} + \frac{\lambda}{2}||d^{k}- \psi(w)-b^{k}||_{\Ell{2}}^{2}$
\STATE $d^{k+1}=\arg \min_{d}  \frac{\lambda}{2}||d- \psi(w^{k+1})+b^{k}||_{\Ell{2}}^{2}+||d||_{\Ell{1}}$
\STATE $b^{k+1}_{i}=b^{k}_{i} + \beta_{i} w^{k+1}_{i} - d^{k+1}_{i}$
\STATE $k=k+1$
\ENDWHILE
\STATE \textbf{Output} $\zeta$ and $d^{k}$ where $\zeta$ is defined as in Theorem \ref{THM:ConvCriteria2} using $\epsilon =\frac{2}{(\sqrt{2}+1)^{2} }tol$ and $\tilde{w}=w^{k}$.
\end{algorithmic}
\end{algorithm}
By Theorem \ref{THM:ConvCriteria2} this algorithm ensures that the objective value is within $tol$ of the optimal value.

\subsection{Adaptive Support Split Bregman}\label{SEC:MultilevelSplit}
The first sub-problem in Algorithms \ref{ALGSplitGeneralized} and \ref{ALG:ModifiedSplit} involves solving a $N \times N$ system of equations.   When the number of assets is large completing this step becomes computational expensive.   This is especially true for financial data where the covariance matrix is ill-conditioned and dense.   Thus Algorithms \ref{ALGSplitGeneralized} and \ref{ALG:ModifiedSplit} may be impractical in applications where real-time results are required or computational performance is limited.
\begin{figure}[h]
    \centering
    \includegraphics[width=6in]{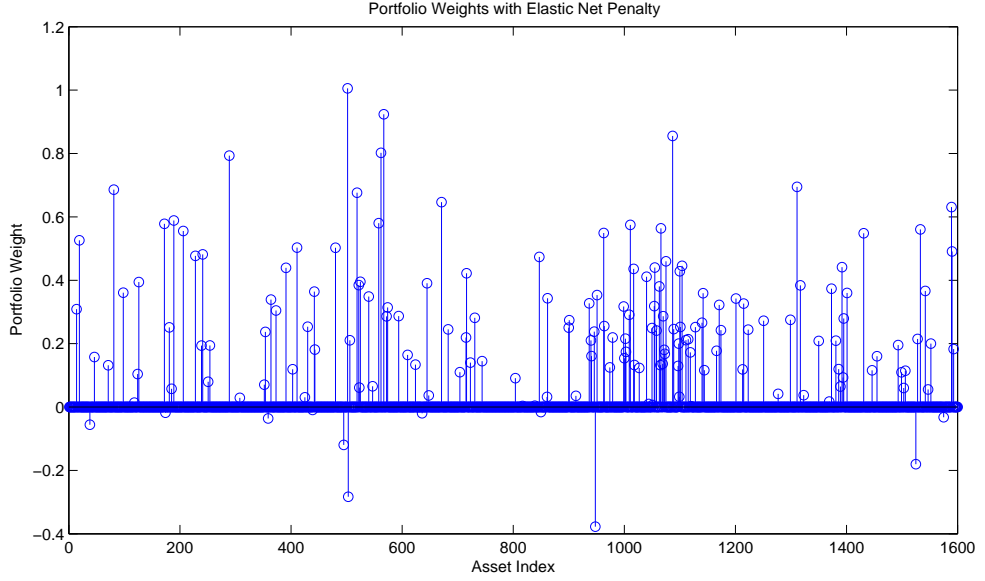}
    \caption{Elastic net penalty promotes sparsity in the portfolio weights}
    \label{fig:Sparse}
\end{figure}

It is well known \cite{SparseStable} that portfolio optimization problems with an $\Ell{1}$ regularization term can result in sparse portfolios i.e. the solution of \eqref{Prog:WtElasticNet} is only non-zero in a small number of indices.   Figure \ref{fig:Sparse} illustrates this behavior by showing the portfolio weights for 1600 assets obtained using the criterion in \ref{Prog:WtElasticNet}.   For this example less than $11 \%$ of the assets have a non-zero weight.

Sparsity of the portfolio weights can be exploited to reduce computational complexity.  To see this suppose $w^{*}$ solves \eqref{Prog:WtElasticNet} and $I=\supp{w^{*}}$ is known a priori (before computing the solution).   Then the problem \eqref{Prog:WtElasticNet} can be relaxed to the equivalent problem
\begin{align*}
  \min_{w}& \;\; w^{T}R_{|I} w-w^{T}\hat{\mu}_{|I}+\EllnW{w}{\vect{\beta}}{1}
\end{align*}
where $R_{|I}$ and $\mu_{|I}$ represent the covariance and mean restricted to $I$.   This problem is of dimension $|I|$ and requires fewer operations to compute per iteration.
This suggests that an Adaptive Support Split-Bregman Algorithm which attempts to solve \eqref{Prog:WtElasticNet} on smaller subspaces, $I$, where $\supp{w^{*}} \subset I$ can save computational time.

To develop an effective algorithm we first derive an optimality condition which can be used as a stopping criterion.
\begin{lemma}
\label{LEMMAstop}
$w^{*}$ solves \eqref{Prog:WtElasticNet} if and only if $|( 2Rw^{*})_{i} - \hat{\mu}_{i}| \le \beta_{i}$ for all $i \not\in \supp{w^{*}}$ and
$(2Rw^{*})_{i} - \hat{\mu}_{i} +\beta_{i}\sign{w^{*}_{i}} =0$ for all $i \in \supp{w^{*}}$.
\end{lemma}
\begin{proof}
Suppose $w^{*}$ solves \eqref{Prog:WtElasticNet} and let $i \in \supp{w^{*}}$.   Then since $w^{*}$ is optimal and $w^{*}_{i} \ne 0$ the partial derivative of the objective function with respect to $w_{i}$ exists and is equal to 0.  Thus
\begin{eqnarray*}
 0&=& \frac{\partial}{\partial w_{i}} \Psi(w)|_{w=w^{*}} \\
  &=& 2(Rw^{*})_{i}  - \hat{\mu}_{i} + \beta_{i}\sign{w^{*}_{i}}.
\end{eqnarray*}
Now suppose $i \notin \supp{w^{*}}$.   Now the partial derivative of the objective function does not exist.  However by optimality we have
\begin{equation*}
0 \in \partial \Psi(w^{*})
\end{equation*}
Thus
\begin{equation*}
\lim_{h \downarrow 0} \frac{\Psi(w^{*}+h \delta_{i}) - \Psi(w^{*})}{h} \ge 0
\end{equation*}
and
\begin{equation*}
\lim_{h \uparrow 0} \frac{\Psi(w^{*}+h \delta_{i}) - \Psi(w^{*})}{h} \le 0
\end{equation*}
which imply
\begin{equation*}
( 2Rw^{*})_{i} - \hat{\mu}_{i} \ge -\beta_{i}
\end{equation*}
and
\begin{equation*}
( 2Rw^{*})_{i} - \hat{\mu}_{i} \le \beta_{i}.
\end{equation*}

For the converse suppose that $|( 2Rw^{*})_{i} - \hat{\mu}_{i}| \le \beta_{i}$ for all $i \not\in \supp{w^{*}}$ and
$(2Rw^{*})_{i} - \hat{\mu}_{i} +\beta_{i}\sign{w^{*}_{i}} =0$ for all $i \in \supp{w^{*}}$.   Choose $\epsilon=\min\{|w_{i}| : i \in \supp{w}\}$. Then for any $w$ such that $||w-w^{*}||_{\infty}<\epsilon$
\begin{eqnarray*}
  \Psi(w)-\Psi(w^{*}) &\ge& \sum_{i \in \supp{w^{*}}} \left( (2Rw^{*})_{i} - \hat{\mu}_{i} +\beta_{i}\sign{w^{*}_{i}}\right)(w_{i}-w^{*}_{i}) + \\
  && \;\;\;\;\;\; +  \sum_{i \not\in \supp{w^{*}}} \left( (2Rw^{*})_{i} - \hat{\mu}_{i} \right)w_{i} + \beta_{i}|w_{i}| \\
  &\ge& 0.
\end{eqnarray*}
Thus $w^{*}$ is locally optimal which implies global optimality.
\end{proof}

Lemma \ref{LEMMAstop} can be used to derive a criterion for determining which indices in a portfolio, $x$, belong in the support.   For example, suppose that $i \not \in \supp{x}$, and $|( 2Rx)_{i} - \hat{\mu}_{i}| > \beta_{i}$.   Then the objective function in \eqref{Prog:WtElasticNet} can be reduced by adding $i$ into $\supp{x}$.   Thus $x$ is not optimal and we should incorporate $i$ into $\supp{x}$.

Next we look at how to prolongate the Split Bregman variables $(w,d,b)$ from a lower dimensional space to a higher dimensional space.  Prolongation of $w$ and $d$ can be achieved through simple zero filling.   Prolongation of $b$ is more delicate.   The following Lemma suggests an effective prolongation.

\begin{lemma}
\label{LEMMAprolong}
Suppose ($w^{*}$,$d^{*}$)  is the solution of \eqref{SBreform1} obtained with Algorithm \ref{ALGSplitGeneralized}.  Then
\begin{equation}\label{EQProlong}
  \lim_{k \rightarrow \infty} b^{k}_{i}  =-(2Rw^{*} - \hat{\mu})_{i}/(\beta_{i} \lambda).
\end{equation}
\end{lemma}
\begin{proof}
By Algorithm \ref{ALGSplitGeneralized} we have for all $k$
\begin{equation*}
   2(Rw^{k+1})_{i}-\hat{\mu}_{i} - \lambda (d^{k} - \psi(w^{k+1}) -b^{k} )_{i}\beta_{i} =0.
\end{equation*}
Since $\lim_{k \rightarrow \infty} w^{k} =w^{*}$ and $\lim_{k \rightarrow \infty} d^{k} =d^{*}$ and $d^{*}=\psi(w^{*}_{i})$ we have
\begin{equation*}
   \lim_{k \rightarrow \infty}   2( Rw^{k+1})_{i}-\hat{\mu}_{i}+  \lambda( b^{k} )_{i}\beta_{i} =0
\end{equation*}
 which implies that
\begin{equation*}
   \lim_{k \rightarrow \infty} (b^{k})_{i} = \frac{\hat{\mu}_{i} - 2(Rw^{*})_{i}}{\beta_{i} \lambda}.
\end{equation*}
\end{proof}

This suggests that the prolongation of $b$ can be defined from equation \eqref{EQProlong}.   For example suppose $(\tilde{w},\tilde{d},\tilde{b})$ solves \eqref{SBreform1} on a restricted domain $I \subset \{1,2, \dots N\}$ and let $w$ and $d$ represent the prolongation of $\tilde{w}$ and $\tilde{d}$ to a set $J \supset I$ i.e.
\begin{equation}
   w_{j}=\begin{cases} \tilde{w}_{j} &\mbox{ if $j \in I$} \\
   0 &\mbox{ if $j \in J-I$ }
   \end{cases}
\end{equation}
\begin{equation}
   d_{j}=\begin{cases} \tilde{d}_{j} &\mbox{ if $j \in I$} \\
   0 &\mbox{ if $j \in J-I$ }.
   \end{cases}
\end{equation}
Then taking a cue from equation \eqref{EQProlong} the prolongation of $\tilde{b}$ may be defined as
\begin{equation}
b_{i}=(-2R_{|J}w  + \hat{\mu}_{|J})_{i}/(\beta_{i} \lambda).
\end{equation}

The Adaptive Support Split Bregman Algorithm for solving \eqref{SBreform1} is given below.
\begin{algorithm}[H]
\caption{Adaptive Support Split Bregman Algorithm for solving \ref{SBreform1}}
\label{ALGSplitMulti}
\begin{algorithmic}
\STATE
\STATE  \textbf{Initialize: }$k=0, w^{0}=0,d^{0}=0,b^{0}=0, \epsilon > 0, M>0$
\STATE  Define $D^{0}=2Rw^{0} - \hat{\mu}$
\WHILE{$|D^{k}_{i}|>\beta_{i}$ for any $i \not\in \supp{w^{k}}$ AND $k < N$}
\STATE  Define the set $J^{k}=\{D^{k}_{i} : i \not\in  \supp{w^{k}}\}$
\STATE  Set $K=M \vee ( k+1 -|\supp{w^{k}}|)$
\STATE  Set $\tilde{J}^{k}$ equal to the largest $K$ elements in $J^{k}$
\STATE  Set $I^{k}=\tilde{J}^{k} \cup \supp{w^{k}}$
\STATE  Run Algorithm \ref{ALG:ModifiedSplit} on $I^{k}$ with initialization $w^{k}_{|I^{k}},b^{k}_{|I^{k}}$, $d^{k}_{I^{k}}$ and tolerance $\epsilon$
\STATE Set $(w^{k+1},d^{k+1})$ to the prolongation of output of previous step
\STATE Set $b^{k+1}_{i}=-2( Rw^{k+1} - \hat{\mu})_{i}/(\beta_{i} \lambda)$,
\STATE Set $D^{k+1}=2Rw^{k+1} - \hat{\mu}$
\STATE $k=k+1$
\ENDWHILE
\end{algorithmic}
\end{algorithm}

The next theorem shows that Algorithm \ref{ALGSplitMulti} converges.
\begin{theorem}
Let $w^{*}$ be the optimal solution to \eqref{Prog:WtElasticNet} and let $w'$ be a solution produced by Algorithm \ref{ALGSplitMulti} for $\epsilon=tol$.   Then
\begin{equation}
  \Psi(w') \le \Psi(w^{*}) + tol.
\end{equation}
\end{theorem}
\begin{proof}
By design the algorithm terminates after at most $N$ iterations.  Suppose the algorithm terminates in $k<N$ iterations.  Let $I^{(k)}$ be the support in iteration $k$ of the Adaptive Support Split-Bregman algorithm.  Then by the proof of Theorem \ref{THM:ConvCriteria2}, $w'$ satisfies the conditions of Theorem \ref{THM:ConvCriteria1} with $\epsilon=tol$.   Thus by Theorem \ref{THM:ConvCriteria1} $\Psi(w') < \Psi(w^{*}) + tol$.   Now suppose the algorithm terminates in $N$ iterations.  Since $I^{(N-1)}$ contains all asset indices it follows by the design of Algorithm \ref{ALG:ModifiedSplit} that
$\Psi(w') < \Psi(w^{*}) + tol$.
\end{proof}

To evaluate the execution speed of Adaptive Support Split-Bregman algorithm we performed a comparison with the following fast algorithms described in the literature:   Split-Bregman algorithm (Algorithm \ref{ALG:ModifiedSplit} ), FISTA \cite{FISTA} and Multilevel Iterated-Shrinkage \cite{MultiLevel}.   To the best of our knowledge these algorithms are considered state of the art for large-scale $\ell_{1}$-penalized quadratic programs.   For the multi-level algorithm proposed in \cite{MultiLevel} we use the FISTA \cite{FISTA} algorithm for all relaxations and lowest level solvers.  To make a fair comparison we have used the same error tolerance of $10^{-6}$ for each algorithm.

Tables \ref{TAB:SpeedBig} and \ref{TAB:SpeedSmall} presents MATLAB run times for solving \eqref{Prog:WtElasticNet} for a large and small basket of US stocks.
The machine running the simulation has the Windows 7 operating system and an Intel i7-3740 processor with 32.0 GB of RAM.

\begin{table}[H]
\scriptsize
\caption{Adaptive Support Split-Bregman converges quickly to a solution for sparse portfolios}
\label{TAB:SpeedBig}
\begin{center}
    \begin{tabular}{| l | l | l | l|l|l|}

    \hline
    Dimension & Sparsity & Adaptive Support  & Split-Bregman & FISTA & Multi-level  \\
    & Level &  Split-Bregman &  &  & FISTA \cite{MultiLevel}  \\ \hline
    2000 & 88 & 0.1 sec & 20.6 sec & 0.4 sec & 0.2 sec \\ \hline
    2000 & 142 & 0.2 sec & 14.5 sec & 0.8 sec &0.2 sec\\ \hline
    2000& 450 & 0.9 sec & 14.6 sec & 3.6 sec & 1.5 sec\\ \hline
    2000 & 853 & 4.8 sec & 23.0 sec & 8.8 sec& 9.2 sec \\ \hline
    2000 & 1692 & 10.4 sec & 38.0 sec & 21.4 sec & 22.7 sec\\ \hline
    3000 & 237 & 0.3 sec & 48.2 sec & 12.9 sec & 2.7 sec\\ \hline
    3000 & 805 & 1.3 sec & 49.9 sec & 55.7 sec & 24.6 sec\\ \hline
    4000 & 234 & 0.5 sec & 107.6 sec & 24.6 sec & 2.2 sec\\ \hline
    \end{tabular}
\end{center}
\end{table}
In Table \ref{TAB:SpeedBig} we see that the Adaptive Support Split-Bregman Algorithm converges much faster than both Split-Bregman, FISTA and Multi-Level FISTA for sparse portfolios taken from a large set of assets.   On the other hand Tables \ref{TAB:SpeedBig} and \ref{TAB:SpeedSmall} show that the advantage of the Adaptive Support Split Bregman algorithm decreases when the cardinality of the asset set is small or when the support of the portfolio is large.
\begin{table}[H]
\scriptsize
\caption{Benefit of Adaptive Support Split-Bregman decreases when dimensionality is small}
\label{TAB:SpeedSmall}
\begin{center}
    \begin{tabular}{| l | l | l | l|l|l|}

    \hline
    Dimension & Sparsity & Adaptive Support  & Split-Bregman & FISTA &Multi-level \\
    & Level &  Split-Bregman &  &  & FISTA \cite{MultiLevel} \\ \hline
     500 & 53 & 0.03 sec & 0.8 sec & 0.02 sec& 0.02 sec \\ \hline
    500 & 150 & 0.09 sec & 0.6 sec & 0.04 sec& 0.03 sec \\ \hline
    500 & 261 & 0.2 sec & 0.5 sec & 0.2 sec& 0.2 sec \\ \hline
    \end{tabular}
\end{center}
\end{table} 
\section{Experimental Results}
In this section we quantify the performance benefit of using a weighted elastic net penalty by testing our criterion in  \eqref{Prog:WtElasticNet} on daily return data from 630 U.S. stocks collected between January 1, 2001 and July 1, 2014 with market capitalization greater than 4 billion US dollars.   The results are then compared with other portfolio selection criteria described in Section \ref{Sec:Crit} and the naive equal-weighted portfolio.

In our experiments we compute new portfolios every 63 trading days using daily returns from the prior 252 trading days as training data for parameter estimation and calibration of the elastic net weights.   Our criteria for evaluating the portfolio performance is the out-of-sample Sharpe ratio of the daily portfolio returns.   Sharpe ratio is defined as the portfolio's excess return divided by its standard deviation.   The formula used for computing the Sharpe ratio is given below
\begin{equation}
  SR=\frac{\frac{1}{\tau}\sum_{i=1}^{\tau}w(t_{i})^{T}r(t_{i})}{\sqrt{\frac{1}{\tau}\sum_{i=1}^{\tau}\left(w(t_{i})^{T}r(t_{i})-\frac{1}{\tau}\left(\sum_{j=1}^{\tau}w(t_{j})^{T}r(t_{j})\right)\right)^{2}}}
\end{equation}
where $\tau$ is the total number of trading days in our 13.5 year data set.   Here $w(t_{i})$ is the portfolio on day $t_{i}$, which is computed from the previous set of training data and remains fixed over intervals of 63 trading days.
\subsection{Parameter Estimation and Calibration}
Due to the large number of assets and small amount of training data, estimation of the covariance and mean in our experiments is performed using shrinkage techniques \cite{SteinsParadox}.  We estimate the covariance matrix using the technique described in \cite{LedoitWolf2004B}.  In that paper the following shrinkage estimator for $\Gamma$ is proposed
\begin{equation}\label{ShrinkCovar}
\hat{\Gamma}=\rho_{1}\Gamma_{S} + \rho_{2}I
\end{equation}
where $\Gamma_{S}$ is the sample average covariance obtained from the training data and where $\rho_{1},\rho_{2}$ are $>0$.   In our experiments we use the optimal values of $\rho_{1}>0$ and $\rho_{2}>0$ which are derived in \cite{LedoitWolf2004B}.  Note that this choice of shrinkage target guarantees that $\hat{\Gamma}$ will be positive definite.

Since the weighted elastic net penalty consists of a squared weighted $\ell_{2}$ norm, the shrinkage in \eqref{ShrinkCovar} may appear to be redundant when applied with the weighted elastic net regularization in Section \ref{sec:WtElastic}.  However, this is not the case since the weights on the weighted elastic net and the shrinkage parameters in \eqref{ShrinkCovar} are adaptively selected according to different criteria.  Thus the covariance shrinkage target becomes a combination of the bootstrap derived target and the target derived according to \cite{LedoitWolf2004B}.   One benefit of this approach is that there will always be some level of $\ell_{2}$ regularization regardless of what the bootstrap criterion derives.

For estimation of the mean we employ a James-Stein estimator \cite{EfronMorris_SteinCompt,JamesSteinQuad} which was proposed for portfolio optimization in \cite{Jorion1986}.   When applying the James-Stein approach we compute the estimate of $\mu$ using the equation
\begin{equation}
\hat{\mu}=(1-\rho)\mu_{S} + \rho \eta\vect{1}.
\end{equation}
Here $\mu_{S}$ is the sample mean vector and $\eta$ is the maximum of average of the sample means and the daily historical return of the US stock market between 1928 and 2000 \cite{StockReturns}
\begin{equation}
\eta= \left( \frac{1}{N}\sum_{i=1}^{N}\mu_{S,i}\right) \vee 0.0004 .
\end{equation}
The value of $\rho$ is set according to \cite{Jorion1986} as
\begin{equation}
\rho=\min\left\{1, \frac{(N-2)}{T_{train}(\mu_{S} -\eta\vect{1})^{T}\hat{\Gamma}^{-1}(\mu_{S} -\eta\vect{1})}\right\}.
\end{equation}

The weights for the weighted elastic net penalty are calibrated using the bootstrap technique described in Section \ref{SEC:Cal} with identical estimation risk aversion factors for mean and squared volatility i.e. $p_{1}=p_{2}$.   Calibration of the weighted LASSO penalty is performed using the technique described in \cite{OptimalSparse}.    Since the weighted LASSO calibration in \cite{OptimalSparse} is only defined up to a constant we perform a parametric study for various constants.   Calibration of the elastic net penalty is handled using the technique described in section 1.6.2 of \cite{L1L2_london}.  The calibration method in \cite{L1L2_london} only determines the sum $\lambda_{1} +\lambda_{2}$ in \eqref{eq:ElasticEq}, the relative weighting of $\lambda_{1}$ and $\lambda_{2}$ is not addressed.   Thus we perform a parametric analysis over the relative weighting between the parameters $\lambda_{1}$ and $\lambda_{2}$ in the elastic net.   For SCAD there are no known calibration methods.   Hence for SCAD we perform a parametric study for various $\lambda$ values and a fixed $a_{SCAD}$ parameter of 3.7 as suggested in \cite{SCAD}.

\subsection{Sharpe Ratio performance}\label{sec:SRperform}
In this section we present performance results for the following 5 mean-variance criteria: 1) unpenalized 2) weighted elastic net penalized, 3) weighted LASSO penalized \cite{OptimalSparse}, 4) elastic net penalized \cite{YenYen}, and 5) SCAD penalized.   As a comparison case we also tested the $1/N$ equal weighted portfolio.
\begin{figure}[h]
    \centering
    \includegraphics[width=5in]{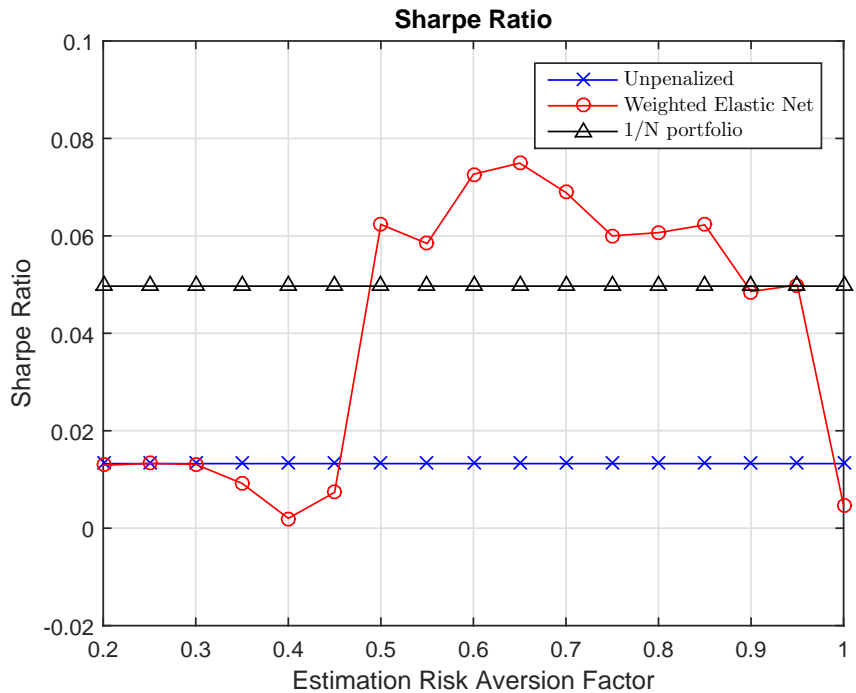}
    \caption{Weighted Elastic Net performance as a function of the estimation risk aversion factor}
    \label{fig:WtElasNet}
\end{figure}

In Figure \ref{fig:WtElasNet} we examine the Sharpe ratios of the weighted elastic net penalty as a function of estimation risk aversion factor, i.e. bootstrap percentile.   As a comparison the performance of the $1/N$ and unpenalized portfolio are also shown.  The figure demonstrates that the weighted elastic net penalized criterion and bootstrap calibration improves Sharpe ratio performance over the $1/N$ and unpenalized portfolio when the estimation risk aversion factor is between 0.5 and 0.95.   Outside of this interval the weighted elastic net penalty did not improve performance, which suggests that a moderate amount of estimation risk aversion is optimal.

\begin{figure}[h]
    \centering
    \includegraphics[width=5in]{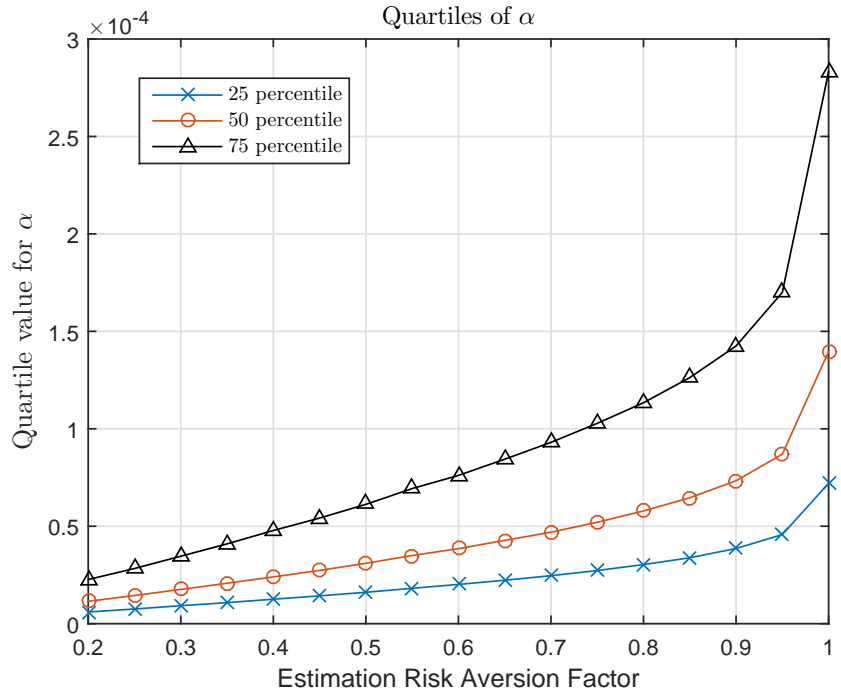}
    \caption{Quartiles for $\alpha$ weights as a function of the estimation risk aversion factor}
    \label{fig:QuartileAlpha}
\end{figure}

\begin{figure}[h]
    \centering
    \includegraphics[width=5in]{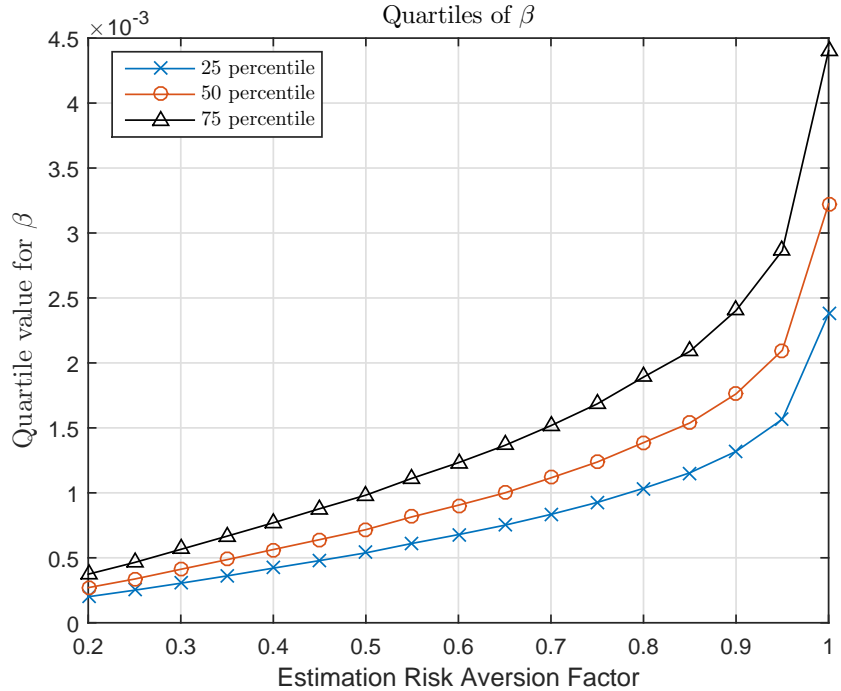}
    \caption{Quartiles for $\beta$ weights as a function of the estimation risk aversion factor}
    \label{fig:QuartileBeta}
\end{figure}

In Figures \ref{fig:QuartileAlpha} and \ref{fig:QuartileBeta} we present the quartiles of the $\alpha$ and $\beta$ parameters obtained from our bootstrap technique as a function of the estimation risk aversion factors.   The values increase sharply when moving from an aversion factor of 0.95 to 1.0.   This may explain the dramatic loss in performance from 0.95 to 1.0 in Figure \ref{fig:WtElasNet}.

For comparison purposes the Sharpe ratio of the weighted LASSO, elastic net and SCAD penalized portfolios are shown in Figures \ref{fig:WtLASSO}, \ref{fig:Elas}, and \ref{fig:SCAD} as a function of their respective penalty scaling parameter.   We see that both weighted LASSO and the elastic net do not perform as well as the weighted elastic net penalty.   This could be a consequence of their calibration being derived from a minimum variance perspective.   The SCAD penalized portfolio performs comparable to the weighted elastic net penalty if the $\lambda$ parameter is chosen correctly.   However, it is still an open question on how to automate the selection of an optimal $\lambda$ in the SCAD penalty for portfolio optimization problems.

\begin{figure}[h]
    \centering
    \includegraphics[width=5in]{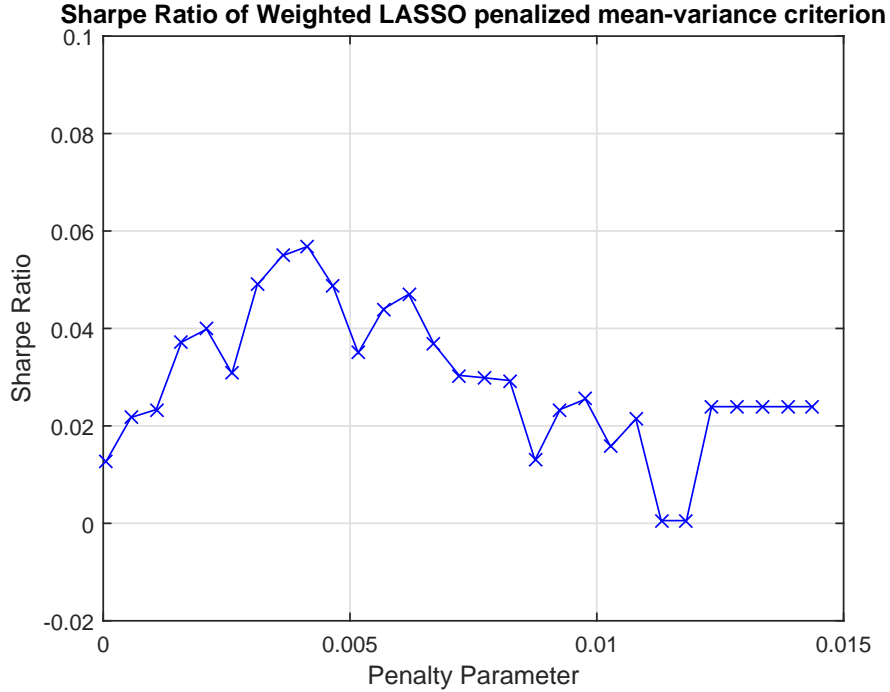}
    \caption{Weighted LASSO performance as a function of penalty normalization factor.   Calibration of relative weight values performed using technique in \cite{OptimalSparse}.}
    \label{fig:WtLASSO}
\end{figure}
\begin{figure}[h]
    \centering
    \includegraphics[width=5in]{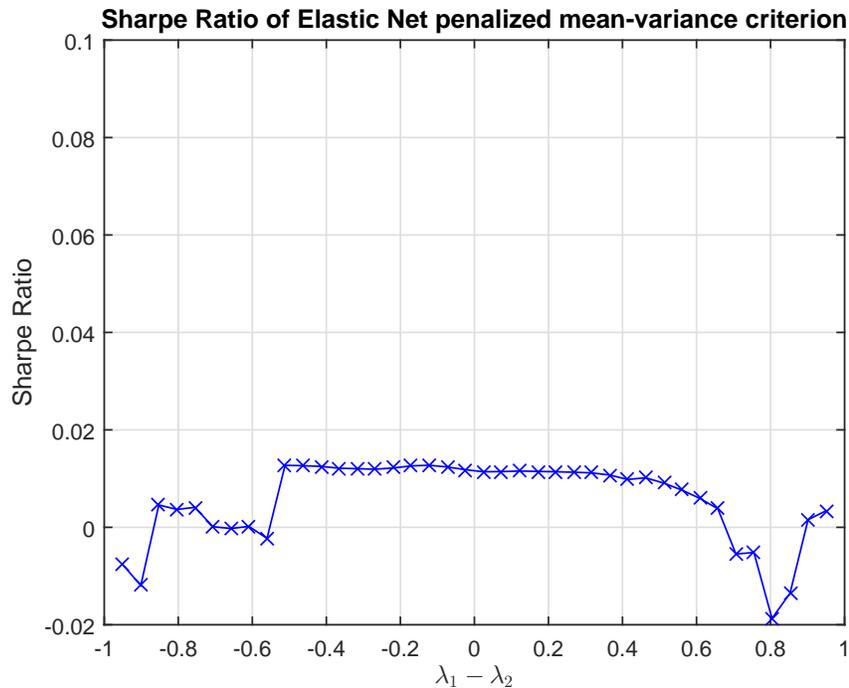}
    \caption{Elastic Net performance as a function of difference of the $\ell_{1}$ and squared $\ell_{2}$ weights.   Calibration performed using the technique in
    \cite{L1L2_london}}
    \label{fig:Elas}
\end{figure}
\begin{figure}[h]
    \centering
    \includegraphics[width=5in]{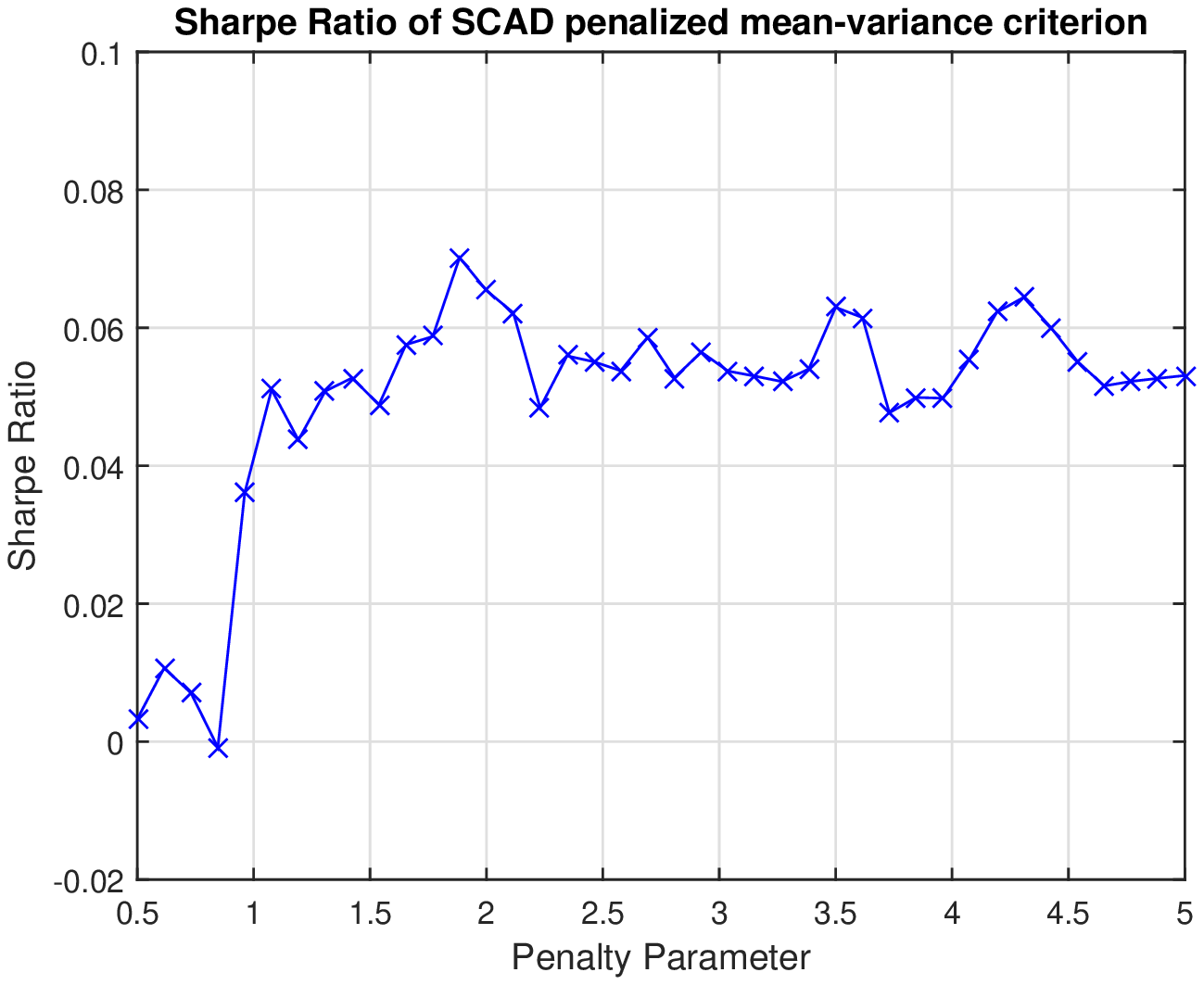}
    \caption{SCAD performance as a function of $\lambda$ parameter.   $a_{SCAD}$ parameter fixed to 3.7}
    \label{fig:SCAD}
\end{figure}

\section{Conclusions and Generalizations}
In this paper the addition of a weighted elastic net penalty to mean-variance objective function has been proposed in order to improve out-of-sample portfolio performance when parameter estimates are uncertain.   We have shown that this approach can be motivated by reformulating the mean-variance criterion as a robust optimization problem.   With this view we develop a data-driven criterion for calibration of the elastic net weights based on bootstrapping and an investor's aversion to model estimation risk.   To compute the portfolio weights efficiently we proposed a novel Adaptive Support Split-Bregman algorithm for solving our proposed optimization criterion.   This technique exploits the sparsity promoting properties of the weighted elastic net penalty to reduce computational requirements.

Our experimental results demonstrate that using the weighted elastic net penalty and calibration approach can result in higher out-of-sample Sharpe ratio than the other norm penalization techniques designed for minimum variance portfolios.  In addition, our MATLAB run-time results indicate that the proposed Adaptive Support Split-Bregman algorithm significantly reduces computation time compared with other algorithms such as Split-Bregman and FISTA.

An interesting question raised by this paper is whether the more general pairwise elastic net penalty in \eqref{Prog:Pairwise} will provide further performance enhancement than the weighted elastic net penalty.   The pairwise penalty appears promising since it is derived from a more flexible model where uncertainty in the off-diagonal of $\Gamma$ is allowed.   However the pairwise elastic net requires specification of up to $\frac{N(N-1)}{2}$ more uncertainty parameters than the weighted elastic net.   In addition numerical algorithms for computing solutions to \eqref{Prog:Pairwise} have not been extensively reported on in the literature.  We plan to investigate these questions in future work.

\section{Acknowledgements}
We would like to thank the reviewers for their valuable comments on our paper.   Their comments
led us to improve the quality of our work.
\appendix
\section{Proofs of Theorems \ref{THM:ConvCriteria1} and \ref{THM:ConvCriteria2}}
In this section we provide proofs for Theorems \ref{THM:ConvCriteria1} and \ref{THM:ConvCriteria2}.   To facilitate the proof we will first reformulate the criterion in \eqref{Prog:WtElasticNet} as a quadratic program.
\subsection{Quadratic Program Reformulation}
Problem \eqref{Prog:WtElasticNet} can be reformulated as a quadratic program with linear inequality constraints by introducing an auxiliary variable $d$,
\begin{align}\label{Prog:QuadReform}
  \min_{w,d}& \;\; \Phi(w,d)  \\
  \textnormal{s.t.   }& -d_{i} \le w_{i} \nonumber \\
  \textnormal{       }& -d_{i} \le -w_{i} \nonumber
\end{align}
where $\Phi(w,d)=w^{T}R w -w^{T}\hat{\mu}+\sum_{i=1}^{N}\beta_{i}d_{i}$ and where $R=\hat{\Gamma}+ D_{\alpha}$.   The Lagrangian for this problem
\begin{equation}\label{EQ:lagrangianQuad}
  L(w,d,\lambda)=w^{T}Rw -w^{T}\hat{\mu}+\sum_{i=1}^{N}\beta_{i}d_{i}  + \sum_{i=1}^{N}\lambda_{i}(-d_{i}-w_{i}) + \sum_{i=1}^{N}\lambda_{i+N}(-d_{i}+w_{i})
\end{equation}
plays an important role in our subsequent analysis in the next section.

\subsection{Approximate Optimality Proofs}
Here we prove Theorems \ref{THM:ConvCriteria1} and \ref{THM:ConvCriteria2} using the quadratic program reformulation \eqref{Prog:QuadReform}.
Our first task is to derive a lower bound on the Lagrangian for a fixed $\lambda$ and when $d=|w|$.   First note that
$R$ is symmetric positive definite whose smallest eigenvalue is $\ge \alpha_{o}$ where
\begin{equation*}
  \alpha_{o}=\min \left\{\alpha_{i}: 1 \le i \le N \right\}.
\end{equation*}
Thus for $d_{i}=|w_{i}|$,$\tilde{d_{i}}=|\tilde{w_{i}}|$ and $\lambda>0$ we have
\begin{eqnarray}\label{EQ:strongConvex}
  \Phi(w,d) &\ge& L(w,d,\lambda) \nonumber \\
   &=&  L(\tilde{w},\tilde{d},\lambda)+\nabla_{w}L(\tilde{w},\tilde{d},\lambda)^{T}(w-\tilde{w})\nonumber \\
  & & \;\;\;\;\;\;\;\;\;\;+ \nabla_{d}L(\tilde{w},\tilde{d},\lambda)^{T}(d-\tilde{d}) + (w-\tilde{w})^{T}H_{w}(\tilde{w},\tilde{d},\lambda)(w-\tilde{w})\nonumber \\
  &\ge& L(\tilde{w},\tilde{d},\lambda)+ \nabla_{w}L(\tilde{w},\tilde{d},\lambda)^{T}(w-\tilde{w})+ \nabla_{d}L(\tilde{w},\tilde{d},\lambda)^{T}(d-\tilde{d}) \nonumber \\
  & & \;\;\;\;\;\;\;\;\;\;\; +  \alpha_{o} \Elln{w-\tilde{w}}{2}^{2} \nonumber \\
  &\ge& L(\tilde{w},\tilde{d},\lambda)+ \nabla_{w}L(\tilde{w},\tilde{d},\lambda)^{T}(w-\tilde{w})+ \nabla_{d}L(\tilde{w},\tilde{d},\lambda)^{T}(d-\tilde{d})\nonumber \\
  & & \;\;\;\;\;\;\;\;\;\;\; + \frac{1}{2} \alpha_{o} \Elln{w-\tilde{w}}{2}^{2} + \frac{1}{2} \alpha_{o} \Elln{d-\tilde{d}}{2}^{2}
 \end{eqnarray}
 where $H_{w}$ is the Hessian of $L$ w.r.t to the $w$ variables.

 We now present two lemmas which will be useful in deriving a stopping criterion.   Our first lemma gives an upper bound for $L$ when the gradient of $L$ is small.
 \begin{lemma}\label{LEM:lemmaClose}
 Suppose $d_{i}=|w_{i}|$ for all $i$ and $\Elln{ \nabla_{w,d}L(\tilde{w},\tilde{d},\lambda)}{2} \le \sqrt{2\epsilon \alpha_{o}}$.   Then $L(\tilde{w},\tilde{d},\lambda) \le \Phi(w^{*},d^{*}) + \epsilon $ where $w^{*}$ solves \eqref{Prog:WtElasticNet} and $d^{*}_{i}=|w^{*}_{i}|$ for all $i$.
 \end{lemma}
 \begin{proof}
 By equation \eqref{EQ:strongConvex} we have
  \begin{eqnarray*}
  \Phi(w^{*},d^{*}) \ge L(w^{*},d^{*},\lambda) &\ge& L(\tilde{w},\tilde{d},\lambda)+ \nabla_{w}L(\tilde{w},\tilde{d})^{T}(w^{*}-\tilde{w})+ \nabla_{d}L(\tilde{w},\tilde{d})^{T}(d^{*}-\tilde{d}) \\
  & & \;\;\;\;\;\;\;\;\;\; + \frac{1}{2}\alpha_{o} \Elln{w^{*}-\tilde{w}}{2}^{2} + \frac{1}{2} \alpha_{o} \Elln{d^{*}-\tilde{d}}{2}^{2}.
 \end{eqnarray*}
 The righthand side is minimized by substituting $-\frac{1}{\alpha_{o}}\nabla_{d}L(\tilde{w},\tilde{d},\lambda)$ in for $(d^{*}-\tilde{d})$ and $-\frac{1}{\alpha_{o}}\nabla_{w}L(\tilde{w},\tilde{d},\lambda)$ in for $(w^{*}-\tilde{w})$.  With these substitutions we obtain
 \begin{eqnarray*}
 \Phi(w^{*},d^{*}) &\ge& L(\tilde{w},\tilde{d},\lambda)-\frac{1}{2 \alpha_{o}} \Elln{ \nabla_{w,d}L(\tilde{w},\tilde{d},\lambda)}{2}^{2} \\
  &\ge&  L(\tilde{w},\tilde{d},\lambda)-\epsilon.
 \end{eqnarray*}
  \end{proof}
The next lemma can be verified easily.
\begin{lemma}
\label{LEM:2x2}
Suppose $|a| \le b$.   Then there exist $x_{1},x_{2} \ge 0$ such that
\begin{eqnarray*}
  x_{1}+x_{2} &=&  b \\
  -x_{1}+x_{2} &=& a.
\end{eqnarray*}
\end{lemma}
\subsubsection{Proof of  Theorem \ref{THM:ConvCriteria1}}
We are now ready to prove Theorem \ref{THM:ConvCriteria1} which establishes a condition for approximate optimality of a portfolio under the weighted elastic net criterion \eqref{Prog:WtElasticNet}.
\begin{proof} {\emph{of Theorem \ref{THM:ConvCriteria1}}} \\
Choose $d^{*}$ and $\tilde{d}$ such that $d^{*}_{i}=|w^{*}_{i}|$ and $\tilde{d}_{i}=|\tilde{w}_{i}|$.   For $i \in \supp{\tilde{w}}$ define $\lambda$ such that
 \begin{equation*}
   \lambda_{i}=\begin{cases}
   0 \mbox{ if } w_{i}>0, i \in \supp{\tilde{w}} \\
   \beta_{i} \mbox { if } w_{i}<0, i \in \supp{\tilde{w}} \\
   \end{cases}
 \end{equation*}
 and for $i \in \supp{\tilde{w}}$,define $\lambda_{i+N}=\beta_{i}-\lambda_{i}$.

 For $i \nin \supp{\tilde{w}}$ we want to define $\lambda_{i}$ and $\lambda_{i+N}$ such that $\lambda_{i} \ge 0$, $\lambda_{i+N} \ge 0$,
 \begin{equation}\label{EQ:lamPluslam}
   \lambda_{i}+ \lambda_{i+N} = \beta_{i}
 \end{equation}
and
 \begin{equation}\label{EQ:lamMinuslam}
-\lambda_{i}+ \lambda_{i+N}= -\frac{\partial}{\partial w_{i}} \big( w^{T}Rw  -w^{T}\hat{\mu} \big)\Big\mid_{w=\tilde{w}}.
\end{equation}
By Lemma \ref{LEM:2x2}, equation \eqref{EQ:ABSineq} implies that such a $\lambda_{i},\lambda_{i+N}$ exists.

Let us form the Lagrangian $L(w,d,\lambda)$ as in equation \eqref{EQ:lagrangianQuad}.   Then for $i \in \supp{\tilde{w}}$
\begin{equation*}
  \partialDer{w_{i}} L(w,d,\lambda)_{|(\tilde{w},\tilde{d})} = \frac{\partial}{\partial w_{i}} \big( w^{T}Rw  -w^{T}\hat{\mu} +\EllnW{w}{\vect{\beta}}{1} \big)\Big\mid_{w=\tilde{w}}
\end{equation*}
and
\begin{equation*}
\partialDer{d_{i}} L(w,d,\lambda)_{|(\tilde{w},\tilde{d})}=0.
\end{equation*}
For $i \nin \supp{\tilde{w}}$ we have by equation \eqref{EQ:lamMinuslam}
\begin{equation*}
  \partialDer{w_{i}} L(w,d,\lambda)_{|(\tilde{w},\tilde{d})} = 0
\end{equation*}
and by equation \eqref{EQ:lamPluslam}
\begin{equation*}
  \partialDer{d_{i}} L(w,d,\lambda)_{|(\tilde{w},\tilde{d})} = 0.
\end{equation*}
It then follows from equation \eqref{EQ:gradCondition} that
\begin{equation*}
\Elln{ \nabla_{w,d}L(\tilde{w},\tilde{d},\lambda)}{2} \le \sqrt{2\epsilon \alpha_{o}}
\end{equation*}
and so by Lemma \ref{LEM:lemmaClose} and our choice of $\lambda$ we have that
\begin{eqnarray*}
\Phi(\tilde{w},\tilde{d}) &=& L(\tilde{w},\tilde{d},\lambda) \\
&\le & \Phi(w^{*},d^{*}) + \epsilon.
\end{eqnarray*}
This clearly implies that
\begin{eqnarray*}
\Psi(\tilde{w}) \le  \Psi(w^{*}) + \epsilon.
\end{eqnarray*}
\end{proof}
\subsubsection{Proof of  Theorem \ref{THM:ConvCriteria2}}
Now we prove Theorem \ref{THM:ConvCriteria2} which can be used to establish a more practical convergence criterion than Theorem \ref{THM:ConvCriteria1}.
\begin{proof}{\emph{of Theorem \ref{THM:ConvCriteria2}}}\\
By construction $\Elln{\zeta-\tilde{w}}{\infty} \le \Elln{\zeta-\tilde{w}}{2} \le \frac{\epsilon\wedge \sqrt{\epsilon \alpha_{o}}}{M}$.  It follows that
 \begin{equation*}
    \sum_{i \in \supp{\zeta}}\left(\frac{\partial}{\partial w_{i}} \big( w^{T}Rw  -w^{T}\hat{\mu} +\EllnW{w}{\vect{\beta}}{1} \big)\Big\mid_{w=\zeta}\right)^{2} \le (\sqrt{2}+1)^{2}\alpha_{o}\epsilon
\end{equation*}
and
\begin{equation*}
-\beta_{i} \le \frac{\partial}{\partial w_{i}} \big( w^{T}Rw  -w^{T}\hat{\mu} \big)\Big\mid_{w=\zeta} \le \beta_{i}
\end{equation*}
for all $i \nin \supp{\zeta}$.
So by Theorem \ref{THM:ConvCriteria1} we have that $\zeta$ satisfies \eqref{EQ:PhiIneqPract}.
\end{proof}

\bibliographystyle{siam}
\bibliography{proposal}

\begin{thebibliography}{10}

\bibitem{Barry1974}
{\sc C.~Barry}, {\em {Portfolio Analysis under Uncertain Means,Variances and
  Covariances}}, Journal of Finance, 29 (1974), pp.~515--522.

\bibitem{FISTA}
{\sc A.~Beck and M.~Teboulle}, {\em {A Fast Iterative Shrinkage-Thresholding
  Algorithm for Linear Inverse Problems}}, SIAM Imaging Scienses, 2 (2009),
  pp.~183--202.

\bibitem{BoydConvex}
{\sc S.~Boyd and L.~Vandenberghe}, {\em Convex Optimization}, Cambridge Univ.
  Press, 2004.

\bibitem{SparseStable}
{\sc J.~Brodie, I.~Daubechies, C.~Gunturk, Y.~Wang, and Y.~O.}, {\em {S}parse
  and {S}table {M}arkowitz {P}ortfolios}, PNAS, 106 (2009), pp.~12267--12272.

\bibitem{HeroRobust}
{\sc Y.~Chen, A.~Wiesel, and A.~Hero}, {\em {Robust Shrinkage Estimation of
  High-dimensional Covariance Matrices}}, {IEEE Transactions on Signal
  Processing}, 59 (2011), pp.~4097--4107.

\bibitem{StockReturns}
{\sc A.~Damodaran}, {\em {Annual Returns on Stock, T.Bonds and T.Bills: 1928 -
  Current}}.
\newblock http://www.stern.nyu.edu/~adamodar/pc/datasets/histretSP.xls.

\bibitem{DeMiguelPORTNORM}
{\sc V.~DeMiguel, L.~Garlappi, F.~Nogales, and R.~Uppal}, {\em A generalized
  approach to portfolio optimization: Improving performance by constraining
  portfolio norms}, Management Science,  (2009), pp.~798--812.

\bibitem{DeMiguelOneOverN}
{\sc V.~DeMiguel, L.~Garlappi, and R.~Uppal}, {\em {Optimal Versus Naive
  Diversification: How Inefficient is the 1{/}N Portfolio Strategy?}}, The
  Review of Financial Studies, 22 (2009), pp.~1915--1953.

\bibitem{EfronMorris_SteinCompt}
{\sc B.~Efron and C.~Morris}, {\em {Stein's Estimation Rule and its
  Competitors- An Empirical Bayes Approach}}, Journal of American Statistical
  Association, 68 (1973), pp.~117--130.

\bibitem{SteinsParadox}
{\sc B.~Efron and C.~Morris}, {\em {Stein's Paradox in Statistics}}, Scientific
  American, 237 (1977), pp.~119--127.

\bibitem{EfronBoot}
{\sc B.~Efron and R.~Tibshirani}, {\em Bootstrap methods for standard
  errors,confidence intervals, and other measures of statistical accuracy},
  Statistical Science, 1 (1986), pp.~54--75.

\bibitem{FamaFrench1993}
{\sc E.~Fama and K.~French}, {\em {Common Risk Factors in the Returns on Stocks
  and Bonds}}, Journal of Financial Economics, 33 (1993), pp.~3--56.

\bibitem{SCAD}
{\sc J.~Fan and R.~Li}, {\em Variable selection via nonconcave penalized
  likelihood and its oracle properties}, Journal of the American Statistical
  Association, 96 (2001), pp.~1348--1360.

\bibitem{VastPortfolio}
{\sc J.~Fan, J.~Zhang, and K.~Yu}, {\em {Vast Portfolio Selection With
  Gross-Exposure Constraints}}, Journal of the American Statistical
  Association, 107 (2012), pp.~592--607.

\bibitem{OptimalSparse}
{\sc B.~Fastrich, S.~Paterlini, and P.~Winker}, {\em Constructing optimal
  sparse portfolios using regularization methods}, Computational Management
  Science, 12 (2014), pp.~417--434.

\bibitem{Frost1986}
{\sc P.~Frost and J.~Savarino}, {\em {An Empirical Bayes Approach to Efficient
  Portfolio Selection}}, The Journal of Financial and Quantitative Analysis, 21
  (1986), pp.~293--305.

\bibitem{GoldFarbIyengar}
{\sc D.~Goldfarb and G.~Iyengar}, {\em {Robust Portfolio Selection Problems}},
  Mathematics of Operations Research, 28 (2003), pp.~1--38.

\bibitem{GoldStein1}
{\sc T.~Goldstein and S.~Osher}, {\em {T}he {S}plit {B}regman {M}ethod for
  $\ell_{1}$ {R}egularized problems}, SIAM Journal on Imaging Sciences, 2
  (2009), pp.~323--343.

\bibitem{InteriorSPD}
{\sc B.~Halldorsson and R.~Tutuncu}, {\em {An Interior Point Method for a Class
  of Saddle Point Problems}}, Journal of Optimization Theory and Applications,
  116 (2003), pp.~559--590.

\bibitem{Ma}
{\sc R.~Jagannathan and T.~Ma}, {\em Risk reduction in large portfolios: Why
  imposing the wrong constraints helps}, Journal of Finance, 58 (2003),
  pp.~1651--1684.

\bibitem{JamesSteinQuad}
{\sc W.~James and C.~Stein}, {\em Estimation with quadratic loss}, in
  Proceedings of the fourth Berkeley symposium on mathematical statistics and
  probability, 1961, pp.~361--379.

\bibitem{Jobson1980}
{\sc J.~Jobson and B.~Korkie}, {\em {Estimation for Markowitz Efficient
  Portfolios}}, Journal of the American Statistical Association, 75 (1980),
  pp.~544--554.

\bibitem{Jorion1986}
{\sc P.~Jorion}, {\em {Bayes-Stein Estimation for Portfolio Analyis}}, The
  Journal of Financial and Quantitative Analysis, 21 (1986), pp.~279--292.

\bibitem{LedoitWolf2004B}
{\sc O.~Ledoit and M.~Wolf}, {\em {A well-conditioned estimator for
  large-dimensional covariance matrices }}, Journal of Multivariate Analysis,
  88 (2004), pp.~365--411.

\bibitem{LedoitWolf2004}
\leavevmode\vrule height 2pt depth -1.6pt width 23pt, {\em {Honey, I Shrunk the
  Sample Covariance Matrix}}, The Journal of Portfolio Management, 30 (2004),
  pp.~110--119.

\bibitem{SparseStableParameter}
{\sc J.~Li}, {\em Sparse and stable portfolio selection with parameter
  uncertainty}, Journal of Business and Economic Statistics, To Appear.

\bibitem{LittermanCovar}
{\sc R.~Litterman and K.~Winkelmann}, {\em {Estimating Covariance Matrices}},
  Risk Management Series,  (1998).
\newblock Goldman Sachs.

\bibitem{PairwiseElastic}
{\sc A.~Lorert, D.~Eis, V.~Kostina, and P.~Ramadge}, {\em Exploiting covariate
  similarity in sparse regression via the pairwise elastic net}, in Proceedings
  of the 13th International Conference on Artificial Intelligence and
  Statistics, vol.~9, Chia Laguna Resort, Sardinia, Italy., 2010.

\bibitem{Markowitz}
{\sc H.~Markowitz}, {\em Portfolio selection}, Journal of Finance, 7 (1952),
  pp.~77--91.

\bibitem{MertonMean}
{\sc R.~Merton}, {\em {On estimating the expected return on the market: An
  exploratory investigation}}, Journal of Financial Economics, 8 (1980),
  pp.~323--361.

\bibitem{MichaudEnigma}
{\sc R.~Michaud}, {\em {The Markowitz optimization enigma: Is optimized
  optimal?}}, Finan, Anal. J., 45 (1989), pp.~31--42.

\bibitem{MichaudBook}
\leavevmode\vrule height 2pt depth -1.6pt width 23pt, {\em {Efficient Assset
  Management: A Practial Guide to Stock Portfolio Optimization}}, Oxford Univ.
  Press, New York, 1993.

\bibitem{NocedalBook}
{\sc J.~Nocedal and S.~Wright}, {\em Numerical Optimization}, Springer, 1999.

\bibitem{Sharpe1963}
{\sc W.~Sharpe}, {\em {A Simplified Model for Portfolio Analysis}}, Management
  Science, 9 (1963), pp.~277--293.

\bibitem{PalomarTyler}
{\sc Y.~Sun, P.~Babu, and D.~Palomar}, {\em Regularized tyler’s scatter
  estimator: Existence, uniqueness, and algorithms}, IEEE Transactions on
  Signal Processing, 62 (2014), pp.~5143--5156.

\bibitem{TutKoenig}
{\sc
  R.~T$\ddot{\textnormal{u}}$t$\ddot{\textnormal{u}}$nc$\ddot{\textnormal{u}}$
  and M.~Koenig}, {\em {Robust asset allocation}}, Annals of Operations
  Research, 132 (2004), pp.~157--187.

\bibitem{MultiLevel}
{\sc E.~Treister and I.~Yavneh}, {\em {A Multilevel Iterated-Shrinkage Approach
  to $\ell_{1}$ Penalized Least-Squares Minimization}}, {IEEE Trans. Signal
  Proc.}, 60 (2012), pp.~6319--6329.

\bibitem{L1L2_london}
{\sc Y.~Yen}, {\em Three essays in financial econometrics}, PhD thesis,
  University of London, 2012.

\bibitem{YenYen}
{\sc Y.~Yen and T.~Yen}, {\em Solving norm constrained portfolio optimization
  via coordinate-wise descent algorithms}, Computational Statistics and Data
  Analysis, 76 (2014), pp.~737--759.

\bibitem{AdaptiveLasso}
{\sc H.~Zou}, {\em The adaptive lasso and its oracle properties}, Journal of
  the American Statistical Association, 101 (2006), pp.~1418--1429.

\bibitem{AdaptiveElasticNet}
{\sc H.~Zou and H.~Zhang}, {\em {On the Adaptive Elastic-Net with a Diverging
  Number of Parameters}}, The Annals of Statistics, 37 (2009), pp.~1733--1751.

\end{thebibliography}

\end{document}